\documentclass[twoside,leqno,twocolumn]{article} 
\usepackage{ltexpprt} 

\usepackage{latexsym}
\usepackage{amssymb}  
\usepackage{amsmath}
\usepackage{amsfonts}
\usepackage{ifthen}
\usepackage{array}
\usepackage[utf8]{inputenc}

\begin{document}

\title{\Large On Termination of Integer Linear Loops}
\author{
Jo\"el Ouaknine\thanks{Supported by EPSRC.} \\
Department of Computer Science\\ 
Oxford University, UK
\and
João Sousa Pinto\thanks{Supported by the ERC Advanced Grant 321171 (ALGAME) and by EPSRC.} \\
Department of Computer Science\\ 
Oxford University, UK
\and 
James Worrell \\
Department of Computer Science\\ 
Oxford University, UK
} 

\date{}

\maketitle
\begin{abstract} 
A fundamental problem in program verification concerns the 
termination of simple linear loops of the form:
\begin{equation*}
 \mbox{$\boldsymbol{x}\gets \boldsymbol{u}$ ; 
\textit{while} $B\boldsymbol{x} \geq \boldsymbol{c}$ \textit{do} 
$\boldsymbol{x}\leftarrow A\boldsymbol{x}+\boldsymbol{a}$\,,}
\end{equation*}
where $\boldsymbol{x}$ is a vector of variables, $\boldsymbol{u}$,
$\boldsymbol{a}$, and $\boldsymbol{c}$ are integer vectors, and $A$ and
$B$ are integer matrices.  Assuming the matrix $A$ is diagonalisable,
we give a decision procedure for the problem of whether, for all
initial integer vectors $\boldsymbol u$, such a loop terminates.  The
correctness of our algorithm relies on sophisticated tools from
algebraic and analytic number theory, Diophantine geometry, and real
algebraic geometry.

To the best of our knowledge, this is the first substantial advance on
a 10-year-old open problem of Tiwari~\cite{Tiw04} and
Braverman~\cite{Bra06}.
\end{abstract}
   
\section{Introduction}
\label{sec:introduction}
Termination is a fundamental decision problem in program verification.
In particular, termination of programs with linear assignments and
linear conditionals has been extensively studied over the last decade.
This has led to the development of powerful techniques to
prove termination via synthesis of linear ranking
functions~\cite{Ben-AmramG13,BradleyMS05,ChenFM12,ColonS01,PodelskiR04},
many of which have been implemented in software-verification tools, such as
Microsoft's \textsc{Terminator}~\cite{CookPR06}.

A very simple form of linear programs are \emph{simple linear
  loops}, that is, programs of the form
\begin{gather*}
\mathsf{P1:}\  \mbox{$\boldsymbol{x}\gets \boldsymbol{u}$ ; 
\textit{while} $B\boldsymbol{x} \geq \boldsymbol{c}$ \textit{do} 
$\boldsymbol{x}\leftarrow A\boldsymbol{x}+\boldsymbol{a}$,}
\end{gather*}
where $\boldsymbol{x}$ is vector of variables, $\boldsymbol{u}$,
$\boldsymbol{a}$, and $\boldsymbol{c}$ are integer vectors, and $A$
and $B$ are integer matrices of the appropriate dimensions.  Here the
loop guard is a conjunction of linear inequalities and the loop body
consists of a simultaneous affine assignment to $\boldsymbol{x}$.  If
the vectors $\boldsymbol{a}$ and $\boldsymbol{c}$ are both zero then
we say that the loop is \emph{homogeneous}.

Suppose that the vector $\boldsymbol{x}$ has dimension $d$.  We say
that \textsf{P1} \emph{terminates} on a set $S\subseteq \mathbb{R}^d$
if it terminates for all initial vectors $\boldsymbol{u} \in S$.
Tiwari~\cite{Tiw04} gave a procedure to decide whether a given simple
linear loop terminates on $\mathbb{R}^d$.  Later
Braverman~\cite{Bra06} showed decidability of termination on
$\mathbb{Q}^d$.  However the most natural problem from the point of
view of program verification is termination on $\mathbb{Z}^d$.  

While termination on $\mathbb{Z}^d$ reduces to termination on
$\mathbb{Q}^d$ in the homogeneous case (by a straightforward scaling
argument), termination on $\mathbb{Z}^d$ in the general case is stated
as an open problem in~\cite{BGM12,Bra06,Tiw04}.  The main result of
this paper is a procedure to decide termination on $\mathbb{Z}^d$ for
simple linear loops when the assignment matrix $A$ is diagonalisable.
This represents the first substantial progress on this open problem in
over $10$ years.

Termination of more complex linear programs can often be reduced to
termination of simple linear loops (see, e.g.,~\cite{CookPR06}
or~\cite[Section 6]{Tiw04}).  On the other hand, termination becomes
undecidable for mild generalisations of simple linear loops, for
example, allowing the update function in the loop body to be piecewise
linear~\cite{BGM12}.

To prove our main result we focus on \emph{eventual non-termination},
where \textsf{P1} is said to be eventually non-terminating on
$\boldsymbol{u} \in \mathbb{Z}^d$ if, starting from initial value
$\boldsymbol{u}$, after executing the loop body $\boldsymbol{x} \gets
A\boldsymbol{x}+\boldsymbol{a}$ a finite number of times \emph{while
  disregarding the loop guard} we eventually reach a value on which
\textsf{P1} fails to terminate.  Clearly \textsf{P1} fails to
terminate on $\mathbb{Z}^d$ if and only if it is eventually
non-terminating on some $\boldsymbol{u} \in \mathbb{Z}^d$.

Given a simple linear loop we show how to compute a convex
semi-algebraic set $W \subseteq \mathbb{R}^d$ such that the integer
points $\boldsymbol{u} \in W$ are precisely the eventually
non-terminating integer initial values.  Since it is decidable whether
a convex semi-algebraic set contains an integer
point~\cite{KhachiyanP97},\footnote{By contrast, recall that the
  existence of an integer point in an \emph{arbitrary} (i.e., not
  necessarily convex) semi-algebraic set---which is equivalent to
  Hilbert's tenth problem---is well-known to be undecidable.} we can
decide whether an integer linear loop is terminating on
$\mathbb{Z}^d$.

Termination over the set of all integer points is easily seen to be
\textbf{coNP}-hard.  Indeed, if the update function in the loop body
is the identity then the loop is non-terminating if and only if there
is an integer point satisfying the guard.  Thus non-termination
subsumes integer programming, which is \textbf{NP}-hard.  By contrast,
even though not stated explicitly in~\cite{Tiw04} and~\cite{Bra06},
deciding termination on $\mathbb{R}^d$ and $\mathbb{Q}^d$ can be done
in polynomial time.\footnote{This observation relies on the facts that
  one can compute Jordan canonical forms of integer matrices and solve
  instances of linear programming problems with algebraic numbers in
  polynomial time~\cite{Cai94,AdlerB94}.}

While our algorithm for deciding termination requires exponential
space, it should be noted that the procedure actually solves a more
general problem than merely determining the existence of a
non-terminating integer point (or, equivalently, the existence of an
eventually non-terminating integer point).  In fact the algorithm
computes a representation of the set of all eventually non-terminating
integer points.  For reference, the closely related problem of
deciding termination on the integer points in a given convex polytope is
\textbf{EXPSPACE}-hard \cite{BGM12}.

As well as making extensive use of algorithms in real algebraic
geometry, the soundness of our decision procedure relies on powerful
lower bounds in Diophantine approximation that generalise Roth's
Theorem.  (The need for such bounds in the inhomogeneous setting was
conjectured in the discussion in the conclusion of~\cite{Bra06}.)  We
also use classical results in number theory, such as the
Skolem-Mahler-Lech Theorem~\cite{Lec53,Mah35,Sko34} on linear
recurrences.  Crucially the well-known and notorious ineffectiveness
of Roth's Theorem (and its higher-dimensional and $p$-adic
generalisations) and of the Skolem-Mahler-Lech Theorem are not a
problem for deciding \emph{eventual} non-termination, which is key to
our approach.

\subsection{Related Work}

Consider the termination problem for a homogeneous linear loop program
\begin{gather*}
\mathsf{P2:}\  \mbox{$\boldsymbol{x}\gets \boldsymbol{u}$ ;
\textit{while} $B\boldsymbol{x} \geq 0$ \textit{do} $\boldsymbol{x}\leftarrow A\boldsymbol{x}$}
\end{gather*}
on a single initial value $\boldsymbol{u}\in \mathbb{Z}^d$.  Each row
$\boldsymbol{b}^T$ of matrix $B$ corresponds to a loop condition
$\boldsymbol{b}^T\boldsymbol{x} \geq 0$.  For each such condition,
consider the integer sequence $\langle x_n : n \in \mathbb{N} \rangle$
defined by $x_n = \boldsymbol{b}^TA^n \boldsymbol{u}$.  Then
\textsf{P2} fails to terminate on an initial value $\boldsymbol{u}$ if
and only if each such sequence $\langle x_n \rangle$ is
\emph{positive}, i.e., $x_n \geq 0$ for all $n$.  It is not difficult
to show that each sequence $\langle x_n \rangle$ considered above is a
\emph{linear recurrence sequence}, thanks to the Cayley-Hamilton
theorem.  Thus deciding whether a homogeneous linear loop program
terminates on a given initial value is at least as hard as the
\emph{Positivity Problem} for linear recurrence sequences, that is,
the problem of deciding whether a given linear recurrence sequence has
exclusively non-negative terms.

The Positivity Problem has been studied at least as far back as the
1970s~\cite{BG07,HHH06,Liu10,RS94,Sal76}.  Thus far decidability is
known only for sequences satisfying recurrences of order $5$ or less.
It is moreover known that showing decidability at order $6$ will
necessarily entail breakthroughs in transcendental number theory,
specifically significant new results in Diophantine
approximation~\cite{OW14:SODA}.  

The key difference between studying termination of simple linear loops
over $\mathbb{Z}^d$ rather than a single initial value is that the
former problem can be approached through eventual termination.  In
this sense the termination problem is related to the \emph{Ultimate
  Positivity Problem} for linear recurrence sequences, which asks
whether all but finitely many terms of a given sequence are
positive~\cite{OuaknineW14a}.  This allows us to bring to bear powerful
non-effective Diophantine-approximation techniques,
specifically the $S$-units Theorem of Evertse, van der Poorten, and
Schlickewei~\cite{Evertse84,PS82}.  Such tools enable us to obtain
decidability of termination for matrices of arbitrary dimension,
assuming diagonalisability.

The paper~\cite{COW14:SODA} studies higher dimensional versions of
Kannan and Lipton's Orbit Problem~\cite{KL86}.  These can be seen as
versions of the termination problem for linear loops on a fixed
initial value.  That work uses substantially different technology from
that of the current paper, including Baker's Theorem on linear forms
in logarithms~\cite{BW93}, and correspondingly relies on restrictions
on the dimension of data in problem instances to obtain decidability.

Termination of $\mathsf{P1}$ under the assumption that all eigenvalues
of $A$ are real was studied in \cite{Jokers,Jokers1} using spectral
techniques. However, as will become clear throughout the course of
this paper, most of the machinery that we use is needed to tackle the
case where there are both real and complex eigenvalues with the same
absolute value. In the setting of \cite{Jokers,Jokers1}, the set of
eventually non-terminating points is in fact a polytope, which can be
effectively computed resorting only to straightforward linear algebra.

While we use spectral and number-theoretic techniques in this paper,
another well-studied approach for proving termination of linear loops
involves designing linear ranking functions, that is, linear functions
from the state space to a well-founded domain such that each iteration
of the loop strictly decreases the value of the ranking
function. However, this approach is incomplete: it is not hard to
construct an example of a terminating loop which admits no linear
ranking function.  Sound and relatively complete methods for
synthesising linear ranking functions can be found in
\cite{PodelskiR04} and \cite{Ben-AmramG13}. Whether a linear
ranking function exists can be decided in polynomial time when the
state space is $\mathbb{Q}^d$ and is \textbf{coNP}-complete when the
state space is $\mathbb{Z}^d$.

\section{Overview of Main Results}
\label{sec:overview}
The main result of this paper is as follows:
\begin{theorem}
  The termination over the integers of simple linear loops of the form
\begin{gather*}
\mathsf{P1:}\  \mbox{$\boldsymbol{x}\gets \boldsymbol{u}$ ; 
\textit{while} $B\boldsymbol{x} \geq \boldsymbol{c}$ \textit{do} 
$\boldsymbol{x}\leftarrow A\boldsymbol{x}+\boldsymbol{a}$}
\end{gather*}
is decidable using exponential space if $A$ is diagonalisable and
using polynomial space if $A$ has dimension at most $4$.
\label{thm:main}
\end{theorem}
In this section we give a high-level overview of the proof of Theorem~\ref{thm:main}.

Let $f:\mathbb{R}^d\rightarrow \mathbb{R}^d$ be the affine function
$f(\boldsymbol{x})=A\boldsymbol{x}+\boldsymbol a$ computed by the body
of the while loop in \textsf{P1} and $P=\{ \boldsymbol{x} \in
\mathbb{R}^d: B \boldsymbol{x}\geq \boldsymbol{c}\}$ the convex polytope
corresponding to the loop guard.  We define the set of
\emph{non-terminating points} to be
\[ \mathit{NT} = \{ \boldsymbol{u}\in \mathbb{R}^d : \forall n \in
\mathbb{N},\, f^n(\boldsymbol{u}) \in P \} \, .\] Following
Braverman~\cite{Bra06}, we moreover define the set of \emph{eventually
  non-terminating points} to be
\[ \mathit{ENT} = \{ \boldsymbol{u}\in \mathbb{R}^d : \exists n \in 
\mathbb{N},\, f^n(\boldsymbol{u}) \in \mathit{NT} \} \, .\] 
It is easily seen from the above definitions that both \textit{NT} and
\textit{ENT} are convex sets.

By definition, \textsf{P1} is non-terminating on $\mathbb{Z}^d$ if
and only if $\mathit{NT}$ contains an integer point.  It is moreover
clear that $\mathit{NT}$ contains an integer point if and only if
$\mathit{ENT}$ contains an integer point.


Recall that a subset of $\mathbb{R}^d$ is said to be
\emph{semi-algebraic} if it is a Boolean combination of sets of the
form $\{ \boldsymbol{x} \in \mathbb{R}^d : p(\boldsymbol{x})\geq 0\}$,
where $p$ is a polynomial with integer coefficients.  Equivalently the
semi-algebraic sets are those definable by quantifier-free first-order
formulas over the structure $(\mathbb{R},<,+,\cdot,0,1)$.  In fact,
since the first-order theory of the reals admits quantifier
elimination~\cite{Tar51}, the semi-algebraic sets are precisely the
first-order definable sets.

Define $W \subseteq \mathbb{R}^d$ to be a
\emph{non-termination witness set} (or simply a witness set) if it
satisfies the following two properties (where $\mathbb{A}$ denotes the
set of algebraic numbers):
\begin{itemize}
\item[(i)] $W$ is convex and semi-algebraic;
\item[(ii)] $W\cap \mathbb{A}^d=\mathit{ENT}\cap \mathbb{A}^d$.
\end{itemize}

The integer points in a witness set $W$ are precisely the integer
points of \textit{ENT}, and so \textsf{P1} is non-terminating on
$\mathbb{Z}^d$ precisely when $W$ contains an integer point.  Our
approach to solving the termination problem consists in computing a
witness set $W$ for a given program and then using the following
theorem of Khachiyan and Porkolab~\cite{KhachiyanP97} to decide whether $W$
contains an integer point.
\begin{theorem}[Khachiyan and Porkolab]
Let $W\subseteq\mathbb{R}^d$ be a convex semi-algebraic set defined by
polynomials of degree at most $D$ and that can be represented in space
$S$. In that case, if $W\cap\mathbb{Z}^d\neq\emptyset$, then $W$ must
contain an integral point that can be represented in space
$SD^{O(d^4)}$.
\end{theorem}

Our approach does not attempt to characterise the set \textit{ENT}
directly, but rather uses the witness set $W$ as a proxy. However, our
techniques do allow us to establish that
$\overline{\mathit{ENT}}=\overline{W}$, which in particular implies
that $\overline{\mathit{ENT}}$ is semi-algebraic, since the closure of
a semi-algebraic set is semi-algebraic (see the Appendix for
details). A natural question is whether the set $ENT$
itself is semi-algebraic, which we leave as an open problem.

We next describe some restrictions on linear loops that can
be made without loss of generality and that will ease our upcoming
analysis.

We first reduce the problem of computing witness sets in the general
case to the same problem in the homogeneous case.  Note that Program
\textsf{P1} terminates on a given initial value $\boldsymbol{u}\in
\mathbb{Z}^d$ if and only if the homogeneous program \textsf{P3} below
terminates for the same value of $\boldsymbol{u}$:
\begin{align*}
\mathsf{P3:} \boldsymbol{x}\gets 
\begin{pmatrix} \boldsymbol{u}\\ 1\end{pmatrix} \mbox{\textit{while}} \begin{pmatrix}B &-\boldsymbol{c} \end{pmatrix}
\boldsymbol{x} \geq 0 \mbox{ \textit{do} } \boldsymbol{x}\leftarrow 
\begin{pmatrix}A&\boldsymbol{a}\\ 0 & 1
\end{pmatrix} \boldsymbol{x}
\end{align*}

Note that if $A$ is diagonalisable then all 
eigenvalues of
$\begin{pmatrix}A&\boldsymbol{a}\\ 0 & 1
\end{pmatrix}$ are simple, with the possible exception of the
eigenvalue $1$.  (Recall that an eigenvalue is said to be simple if it
has multiplicity one as a root of the minimal polynomial of $A$.)  Now
if $W$ is a witness set for program \textsf{P3} then $\left\{
  \boldsymbol{u} \in \mathbb{R}^d : \begin{pmatrix} \boldsymbol{u}\\
    1\end{pmatrix} \in W \right\}$ is a witness set for \textsf{P1}.
We conclude that, in order to settle the inhomogeneous case with a
diagonalisable matrix, it suffices to compute a witness set in the
case of a homogeneous linear loop \textsf{P2} in which the only
repeated eigenvalues of the new matrix $A$ are positive and real.
Likewise, to handle the inhomogeneous case for matrices of dimension
at most $d$, it suffices to be able to compute witness sets in the
homogeneous case for matrices of dimension at most $d+1$.%
\footnote{Note that whilst Braverman~\cite{Bra06} shows how to decide
  termination over the integers for homogeneous programs with
  arbitrary update matrices, he does \emph{not} compute a witness set
  for such programs---indeed this remains an open problem since it
  would enable one to solve termination over the integers for
  arbitrary inhomogeneous programs.}

We can further simplify the homogeneous case by restricting to loop
guards that comprise a single linear inequality.  To see this, first
note that program \textsf{P2} above is eventually non-terminating on
$\boldsymbol{u}$ if and only if for each row $\boldsymbol{b}^T$ of $B$
program \textsf{P4} below is eventually non-terminating on
$\boldsymbol{u}$:
\begin{gather*}
\mathsf{P4:}\ \mbox{$\boldsymbol{x}\gets \boldsymbol{u}$ ; 
\textit{while} $\boldsymbol{b}^T\boldsymbol{x} \geq 0$ \textit{do} $\boldsymbol{x}\leftarrow A\boldsymbol{x}$.}
\end{gather*}
Noting that the finite intersection of convex semi-algebraic sets is
again convex and semi-algebraic, we can compute a witness set
for \textsf{P2} as the intersection of witness sets for each version
of \textsf{P4}.

The final simplification concerns the notion of
non-degeneracy.  We say that matrix $A$ is \emph{degenerate} if it has
distinct eigenvalues $\lambda_1 \neq \lambda_2$ whose quotient
$\lambda_1/\lambda_2$ is a root of unity.  

Given an arbitrary matrix $A$, let $L$ be the least common multiple of all orders of quotients of distinct eigenvalues of $A$ which are roots
of unity. It is known that $L=2^{O(d\sqrt{\log d})}$ \cite{BOOK}.  The
eigenvalues of the matrix $A^L$ have the form $\lambda^L$ for
$\lambda$ an eigenvalue of $A$, by the spectral mapping theorem.  It follows that
$A^L$ is non-degenerate, since if $\lambda_1,\lambda_2$ are
eigenvalues of $A$ such that $\lambda^L_1/\lambda^L_2$ is a root of
unity then $\lambda_1/\lambda_2$ is a root of unity and hence
$\lambda^L_1/\lambda^L_2=1$. Note that all eigenvectors of $A$ are still eigenvectors of $A^L$, thus $A^L$ will be diagonalisable whenever $A$ is.

Now program \textsf{P4} is eventually non-terminating on
$\boldsymbol{u}\in \mathbb{Z}^d$ if and only if program \textsf{P5}
below is eventually non-terminating on the set
$\{\boldsymbol{u},A\boldsymbol{u}, \ldots,A^{L-1}\boldsymbol{u}\}$:
\begin{gather*}
\mathsf{P5:}\ \mbox{$\boldsymbol{x}\gets \boldsymbol{v}$ ; 
\textit{while} $\boldsymbol b^T\boldsymbol{x} \geq 0$ \textit{do} $\boldsymbol{x}\leftarrow A^L\boldsymbol{x}$.}
\end{gather*}
Thus if $W$ is a witness set for \textsf{P5} then $\bigcap_{i=0}^{L-1}
\{ \boldsymbol{u} \in \mathbb{Z}^d : A^i\boldsymbol{u} \in W\}$ is a witness set for
\textsf{P4}.

The main technical result of the paper is the following proposition:
\begin{proposition}
Given a homogeneous simple linear loop 
\begin{gather*}
\mathsf{P4:}\ \mbox{$\boldsymbol{x}\gets \boldsymbol{u}$ ; 
\textit{while} $\boldsymbol{b}^T\boldsymbol{x} \geq 0$ \textit{do} $\boldsymbol{x}\leftarrow A\boldsymbol{x}$,}
\end{gather*}
such that $A$ is non-degenerate and either $A$ has dimension at most
$5$ or all complex eigenvalues of $A$ are simple,  we can compute a
witness set for \textsf{P4} using exponential space if $A$ is diagonalisable and using polynomial space if $A$ has dimension at most $5$.
\label{prop:main}
\end{proposition}

Bearing in mind that the transformation from $\mathsf{P1}$ to
$\mathsf{P4}$ increases the dimension of $A$ by one and does not
introduce repeated complex eigenvalues, it follows from
Proposition~\ref{prop:main} that we can also compute witness sets for
simple linear loops of the form \textsf{P1}
under the assumptions of Theorem~\ref{thm:main}, and thus we obtain
the decidability part of Theorem~\ref{thm:main}.  The
exponential-space bound in Theorem~\ref{thm:main} is obtained by
bounding the representation of the witness set in
Proposition~\ref{prop:main} (see Section~\ref{sec:complexity}).

In the rest of this section we give a brief summary of the proof of
Proposition~\ref{prop:main}.  

To compute a witness set $W$ for \textsf{P4} we first partition the
eigenvalues of the update matrix $A$ by grouping eigenvalues of equal
modulus.  Correspondingly we write $\mathbb{R}^d$ as a direct sum
$\mathbb{R}^d = V_1 \oplus \ldots \oplus V_m$, where each subspace
$V_i$ is the sum of (generalised) eigenspaces of $A$ associated to
eigenvalues of the same modulus.  Assume that $V_1$ corresponds to the
eigenvalues of maximum modulus, $V_2$ the next greatest modulus,
etc.  Then there are two main steps in the construction of $W$:
\begin{enumerate}
\item 
 By analysing multiplicative relationships among eigenvalues of the
 same modulus, we show that for each subspace $V_i$ the set
 $\mathit{ENT}\cap V_i$ of eventually non-terminating initial values
 in $V_i$ is semi-algebraic.
\item
Given $\boldsymbol{v} \in \mathbb{R}^d$, we can write
$\boldsymbol{v}=\boldsymbol{v}_1+\ldots+\boldsymbol{v}_m$, with
$\boldsymbol{v}_i \in V_i$.  Using Theorem~\ref{thm:s-units} on
$S$-units, we show that if all entries of $\boldsymbol{v}$ are
algebraic numbers then the eventual non-termination of \textsf{P4} on
$\boldsymbol{v}$ is a function of its eventual non-termination on each
$\boldsymbol{v}_i$ separately.  More precisely we look for the first
$\boldsymbol{v}_i$ such that the sequence $\langle
\boldsymbol{b}^TA^n\boldsymbol{v}_i : n \in \mathbb{N}\rangle$ is infinitely often non-zero. Then \textsf{P4} is eventually non-terminating on
$\boldsymbol{v}$ if and only if it is eventually non-terminating on
$\boldsymbol{v}_i$.
\end{enumerate}

The computability of a witness set $W$ easily follows from items 1 and
2 above. Our techniques require that the update matrix in the
original linear loop \textsf{P1} either be diagonalisable or have
dimension at most $4$.  Eliminating these restrictions seems to
require solving the Ultimate Positivity Problem for linear recurrence
sequences of order greater than $5$, which in turn requires solving
hard open problems in the theory of Diophantine approximation
\cite{OW14:SODA}.

\section{Groups of Multiplicative Relations}
\label{sec:mult}
This section introduces some concepts concerning groups of
multiplicative relations among algebraic numbers.  Here we will assume
some basic notions from algebraic number theory and the first-order
theory of reals.  We assume also a natural first-order interpretation of
the field of complex numbers in the ordered field of real numbers (in
which each complex number is encoded as a pair comprising its real and
imaginary parts).  Under this interpretation we refer to sets of
complex numbers as being semi-algebraic and first-order definable.
Details of the relevant notions can be found in the Appendix.

Let $\mathbb{T}=\lbrace z\in\mathbb{C}: \lvert z\rvert =1\rbrace$.  We
define the \emph{$s$-dimensional torus} to be $\mathbb{T}^s$,
considered as a group under componentwise multiplication.

Given a tuple of algebraic numbers
$\boldsymbol\lambda=(\lambda_1,\ldots,\lambda_s)$, in this section we
consider how to effectively represent the \emph{orbit} $\{
\boldsymbol\lambda^n : n \in \mathbb{N}\}$.  More
precisely, we will give an algebraic representation of the topological
closure of that orbit in $\mathbb{T}^s$.

The \emph{group of multiplicative relations} of
$\boldsymbol{\lambda}$, which is an additive subgroup of
$\mathbb{Z}^s$, is defined as
\begin{equation*}
L(\boldsymbol{\lambda})=\lbrace \boldsymbol{v}\in \mathbb{Z}^s : \boldsymbol\lambda^{\boldsymbol v}=1 \rbrace \, ,
\end{equation*}
where $\boldsymbol\lambda^{\boldsymbol v}$ is defined to be
$\lambda_1^{v_1}\cdots\lambda_s^{v_s}$ for $\boldsymbol{v}\in \mathbb{Z}^s$, that is, exponentiation acts
coordinatewise.

Since $\mathbb{Z}^s$ is a free abelian group, its subgroups are also
free.  In particular, $L(\boldsymbol\lambda)$ has a finite basis. The
following powerful theorem of Masser~\cite{Mas88} gives bounds on the
magnitude of the components of such a basis.

\begin{theorem}[Masser]
\label{masser}
The free abelian group $L(\boldsymbol{\lambda})$ has a basis $\boldsymbol{v}_1,\cdots,\boldsymbol{v}_l\in\mathbb{Z}^s$ for which
\[ \max\limits_{1\leq i\leq l,1\leq j\leq s} \lvert v_{i,j} \rvert \leq (D\log H)^{O(s^2)} \]
where $H$ and $D$ bound respectively the heights and degrees of all the $\lambda_i$.
\end{theorem}
Membership of a tuple $\boldsymbol{v}\in \mathbb{Z}^s$ in
$L(\boldsymbol{\lambda})$ can be computed in polynomial space, using a
decision procedure for the existential theory of the reals.  In
combination with Theorem~\ref{masser}, it follows that we can compute
a basis for $L(\boldsymbol{\lambda})$ in polynomial space by
brute-force search.

Corresponding to $L(\boldsymbol{\lambda})$, we consider the following
multiplicative subgroup of $\mathbb{T}^s$:
\begin{equation*}
T(\boldsymbol{\lambda})=\lbrace \boldsymbol \mu\in\mathbb{T}^s : \forall \boldsymbol v\in L(\boldsymbol\lambda),\,\boldsymbol\mu^{\boldsymbol v}=1\rbrace \, .
\end{equation*}
If $V$ is a basis of $L(\boldsymbol{\lambda})$ then we can
equivalently characterise $T(\boldsymbol{\lambda})$ as $\{
\boldsymbol{\mu} \in \mathbb{T}^s: \forall \boldsymbol{v}\in
V,\,\,\boldsymbol\mu^{\boldsymbol v}=1\}$.  Crucially, this finitary
characterisation allows us to represent $T(\boldsymbol\lambda)$ as a
semi-algebraic set.

We will use the following classical lemma of Kronecker on simultaneous
Diophantine approximation, in order to show that the orbit $\lbrace
\boldsymbol\lambda^n : n\in\mathbb{N} \rbrace$ is a dense subset of
$T(\boldsymbol{\lambda})$.

\begin{lemma}
  Let $\boldsymbol \theta,\boldsymbol \psi\in\mathbb{R}^s$. Suppose that for all $\boldsymbol v\in\mathbb{Z}^s$, if
  $\boldsymbol v^T\boldsymbol \theta\in\mathbb{Z}$ then also
  $\boldsymbol v^T\boldsymbol\psi\in\mathbb{Z}$, i.e., all integer
  relations among the coordinates of $\boldsymbol \theta$ also hold
  among those of $\boldsymbol\psi$ (modulo $\mathbb{Z}$). Then, for
  each $\varepsilon>0$, there exist $\boldsymbol p\in\mathbb{Z}^s$ and
  a non-negative integer $n$ such that
\[ \| n\boldsymbol\theta - \boldsymbol p - \boldsymbol\psi \|_\infty \leq\varepsilon \, .\]
\end{lemma}

We now arrive at the main result of the section:

\begin{theorem}
\label{dense}
Let $\boldsymbol{\lambda}\in\mathbb{T}^s$. Then the orbit $\lbrace \boldsymbol\lambda^n : n\in\mathbb{N} \rbrace$ is a dense subset of $T(\boldsymbol{\lambda})$.
\end{theorem}

\begin{proof}
  Let $\boldsymbol \theta\in\mathbb{R}^s$ be such that
  $\boldsymbol\lambda=e^{2\pi i\boldsymbol\theta}$ (with
  exponentiation operating coordinatewise). Notice that
  $\boldsymbol\lambda^{\boldsymbol v}=1$ if and only if $\boldsymbol v^T
  \boldsymbol\theta\in\mathbb{Z}$. If $\boldsymbol\mu\in
  T(\boldsymbol{\lambda})$, we can likewise define
  $\boldsymbol\psi\in\mathbb{R}^s$ to be such that
  $\boldsymbol\mu=e^{2\pi i \boldsymbol\psi}$. Then the premisses of
  Kronecker's lemma apply to $\boldsymbol \theta$ and $\boldsymbol
  \psi$. Thus, given $\varepsilon>0$, there exist a non-negative
  integer $n$ and $\boldsymbol p\in\mathbb{Z}^s$ such that $\|
  n\boldsymbol \theta -\boldsymbol p-\boldsymbol \psi \|_\infty
  \leq\varepsilon$. Whence
\[ \| \boldsymbol\lambda^n-\boldsymbol\mu \|_\infty = \| e^{2\pi i(n\boldsymbol\theta-\boldsymbol p)}-e^{2\pi i \boldsymbol\psi} \|_\infty \leq \]
\[ \| 2\pi (n\boldsymbol \theta -\boldsymbol p - \boldsymbol \psi) \|_\infty \leq 2\pi\varepsilon \, .\]
\end{proof}

\section{Algorithm for Universal Termination}
Our goal in this section is to prove the following proposition, which
is restated from Section~\ref{sec:overview}.  We have already shown in
Section~\ref{sec:overview} that the main result of this paper,
Theorem~\ref{thm:main}, then follows.

\begin{proposition}
Given a homogeneous simple linear loop 
\begin{gather*}
\mathsf{P4:}\ \mbox{$\boldsymbol{x}\gets \boldsymbol{u}$ ; 
\textit{while} $\boldsymbol{b}^T\boldsymbol{x} \geq 0$ \textit{do} $\boldsymbol{x}\leftarrow A\boldsymbol{x}$,}
\end{gather*}
such that $A$ is non-degenerate, we can compute a
witness set for \textsf{P4} using exponential space if all complex eigenvalues of $A$ are simple and using polynomial space if $A$ has dimension at most $5$.
\label{prop:main2}
\end{proposition}


Define the \emph{index} of an eigenvalue of $A$ to be its multiplicity
as a root of the minimal polynomial of $A$. An eigenvalue is said to be \textit{simple} if it has index $1$ and \textit{repeated} otherwise. We can write matrix $A$
in the form $A=P^{-1}JP$ for some invertible matrix $P$ and block
diagonal Jordan matrix $J=\mathit{Diag}(J_1,\ldots,J_N)$, with each
block $J_i$ having the form
\[\begin{pmatrix} \lambda & 1 & 0 & \ldots & 0\\
                     0 & \lambda & 1& \ldots & 0\\
                    \vdots & \vdots & \vdots & \ddots & \vdots\\
                   0 & 0 & 0 & \ddots & 1\\
                   0 & 0 & 0 & \ldots & \lambda
\end{pmatrix} \, ,\]
where $\lambda$ is an eigenvalue of $A$ whose index equals the the dimension of the block.  
  The entries of $P$ 
are all algebraic numbers lying in the 
extension field of $\mathbb{Q}$ generated by the eigenvalues of $A$.

The $n$-th power of the matrix $J$ has the form
$J^n=\mathit{Diag}(J^n_1,\ldots,J^n_N)$, where each block $J_i^n$ has the form
\[\begin{pmatrix} 
\lambda^n & n\lambda^{n-1} & \binom{n}{2}\lambda^{n-2} 
& \ldots & \binom{n}{\nu-1}\lambda^{n-\nu+1}\\
                     
0 & \lambda^n & n\lambda_i^{n-1} & \ldots & 
\binom{n}{\nu-2}\lambda^{n-\nu+2}\\
                    
\vdots & \vdots & \vdots & \ddots & \vdots\\
                   0 & 0 & 0 & \ddots & n\lambda^{n-1}\\
                   0 & 0 & 0 & \ldots & \lambda^n
                 \end{pmatrix} \, ,\] where $\lambda$ is an eigenvalue of $A$ of index $\nu$, and $\binom{n}{k}=0$ if $n<k$.

                 Let $A$ have eigenvalues $\lambda_1,\ldots,\lambda_l$,
                 with respective indices $\nu_1,\ldots,\nu_l$.  
Given $\boldsymbol{u}\in\mathbb{R}^d$, from
                 our observations on the form of $J^n$, we can write
\begin{gather}
\boldsymbol{b}^TA^n\boldsymbol{u} =
\sum_{j=1}^l \sum_{k=0}^{\nu_j-1} 
\boldsymbol{\alpha}_{j,k}^T\boldsymbol{u}\, n^k\lambda_j^n \, ,
\label{eq:master}
\end{gather}
where the $\boldsymbol{\alpha}_{j,k}$ are vectors of algebraic numbers
that do not depend on $\boldsymbol{u}$, and the equation holds for all
$n\geq d$.

Since the characteristic polynomial of $A$ has integer coefficients,
the eigenvalues of $A$ are all algebraic integers.  Moreover, since
for any positive integer $t>0$ we have that $t\cdot
\boldsymbol{b}^TA^n\boldsymbol{u} \geq 0$ if and only if
$\boldsymbol{b}^TA^n\boldsymbol{u}\geq 0$, by rescaling we can assume that
the vectors $\boldsymbol{\alpha}_{j,k}$ in (\ref{eq:master}) are
comprised of algebraic integers.

Now let us partition the eigenvalues of $A$ into sets $S_1,\ldots,S_m$
by grouping eigenvalues of equal modulus.  Assume that $S_1$ contains
eigenvalues of maximum modulus, $S_2$ eigenvalues of the next greatest
modulus, etc.  Correspondingly we write $\mathbb{R}^d$ as a
direct sum of subspaces $\mathbb{R}^d = V_1 \oplus \ldots \oplus V_m$,
where each subspace $V_i$ is the sum of (generalised) eigenspaces of
$A$ associated to eigenvalues in $S_i$.  By the assumption that $A$ is
non-degenerate, i.e., that no quotient of two distinct eigenvalues is
a root of unity, $S_i$ cannot have both a positive and a negative real
eigenvalue of the same modulus.  Thus each set $S_i$ contains at most one
real eigenvalue.

\subsection{Eventual Non-Termination on Subspace $V_i$}
We first consider the eventual non-termination of \textsf{P4} on
initial vectors in the subspace $V_i$ for a fixed $i \in
\{1,\ldots,m\}$.  Writing $\mathit{ENT}_i := \mathit{ENT} \cap V_i$,
our goal is to show that $\mathit{ENT}_i$ is semi-algebraic.

Given $\boldsymbol{u} \in V_i$, membership of $\boldsymbol{u}$ in
$\mathit{ENT}_i$ can be characterised in terms of the \emph{ultimate
  positivity} of the sequence $\langle\,
\boldsymbol{b}^TA^n\boldsymbol{u}:n\in\mathbb{N}\,\rangle$.  More
precisely, $\boldsymbol{u} \in \mathit{ENT}_i$ if and only if
$\boldsymbol{b}^TA^n\boldsymbol{u} \geq 0$ for all but finitely many
$n$.  In particular, defining
\[ \mathit{ZERO}:=\{\boldsymbol{u} \in \mathbb{R}^d: \forall
n\geq d,\, \boldsymbol{b}^TA^n\boldsymbol{u}=0 \} \, \] and
$\mathit{ZERO}_i:=\mathit{ZERO}\cap V_i$, we have that
$\mathit{ZERO}_i\subseteq \mathit{ENT}_i$.

It is easy to see that $\mathit{ZERO}_i$ is semi-algebraic.  Indeed
the uniqueness part of~\cite[Proposition 2.11]{TUCS05} implies that
$\boldsymbol{b}^TA^n\boldsymbol{u}=0$ for all $n\geq d$ if and only if
each term $n^k\lambda_j^n$ has coefficient zero in the expression
(\ref{eq:master}).  Thus
\[ \mathit{ZERO} = \left\{ \boldsymbol{u}\in \mathbb{R}^d :
\bigwedge_{j=1}^{l}\bigwedge_{k=0}^{\nu_j-1}
\boldsymbol{\alpha}_{j,k}^T\boldsymbol{u}= 0 \right\} \, .\] is
semi-algebraic.  Since $V_i$ is a semi-algebraic subset of
$\mathbb{R}^d$, being spanned by a subset of the columns of $P$, it
follows that $\mathit{ZERO}_i$ is semi-algebraic.

\begin{proposition}
The set $\mathit{ENT}_i$ is semi-algebraic for each $i\in \{1,\ldots,m\}$.
\label{prop:semi-alg}
\end{proposition}
\begin{proof}
  We consider three (overlapping) cases.  Under the hypotheses of
  Proposition~\ref{prop:main2} at least one of these cases will apply.

\paragraph{Case I: $A$ has dimension at most $5$.}
Assume that $A$ has dimension at most $5$.  The situations in which
$S_i$ does not contain a positive real eigenvalue, or 
all of the complex eigenvalues in $S_i$ are simple, will be handled
under Cases II and III, below.  Otherwise, let $\lambda\in S_i$ be a complex
eigenvalue of index at least $2$.  Since $A$ has dimension at most
$5$, it must be the case that $\lambda$ and its complex conjugate
$\overline{\lambda}$ both have index exactly $2$.  Let $\rho \in S_i$
be the positive real eigenvalue.  Since
$A$ has dimension at most $5$, $\rho$ must be simple.  Thus
$S_i=\{\rho,\lambda,\overline{\lambda}\}$ contains all the eigenvalues
of $A$.

For $\boldsymbol{u}\in V_i$ we can write 
\begin{align*}
\boldsymbol{b}^TA^n\boldsymbol{u} = 
\big(\boldsymbol{\alpha}_0\rho^n  + (\boldsymbol{\beta}_0+ \boldsymbol{\beta}_1n)\lambda^n 
 + \overline{(\boldsymbol{\beta}_0+\boldsymbol{\beta}_1n)
\lambda^n}\big)^T\boldsymbol{u} \, ,
\label{eq:small}
\end{align*}
for all $n\geq d$, 
where $\boldsymbol{\alpha}_0$ is a vector of real algebraic numbers,
$\boldsymbol{\beta}_0,\boldsymbol{\beta}_1$ are vectors of complex
algebraic numbers.

If $\boldsymbol{\beta}_1^T\boldsymbol{u}\neq 0$, then as $n$ tends to
infinity the dominant terms on the right-hand side above are constant
multiples of $n\lambda^n$ and $n\overline{\lambda^n}$.  In this case
it follows from~\cite[Lemma 4]{Bra06} that
$\boldsymbol{b}^TA^n\boldsymbol{u}$ changes sign infinitely often as
$n$ grows, and hence $\boldsymbol{u}\not\in\mathit{ENT}_i$.

The argument in case $\boldsymbol{\beta}_1^T\boldsymbol{u}=0$ is a
simple version of the approach in Case III, however we include details
since the reader may find this special case instructive.

Define $f:\mathbb{T}\rightarrow\mathbb{R}$ by
\[ f(z)=\boldsymbol{\alpha}_0^T\boldsymbol{u}+
\boldsymbol{\beta}_0^T\boldsymbol{u}z+
\overline{\boldsymbol{\beta}_0^T\boldsymbol{u}z} \, .\]
Then $\boldsymbol{b}^TA^n\boldsymbol{u} = \rho^n f(\lambda^n/\rho^n)$ for all $n\geq d$.

Since $A$ is assumed to be non-degenerate, $\lambda/\rho$ is not a
root of unity.  Thus $\{\lambda^n/\rho^n:n\in\mathbb{N}\}$ is dense in
$\mathbb{T}$.  It follows that $\boldsymbol{u}\in\mathit{ENT}$ if and
only if $f(z)\geq 0$ for all $z\in\mathbb{T}$.  By inspection this
last condition is equivalent to $\boldsymbol{\alpha}_0^T\boldsymbol{u}
\geq 2|\boldsymbol{\beta}_0^T\boldsymbol{u}|$.
We conclude that
\[ \mathit{ENT}_i = \left\{ \boldsymbol{u}\in V_i : 
\boldsymbol{\beta}_1^T\boldsymbol{u}=0 \wedge
\boldsymbol{\alpha}_0^T\boldsymbol{u} \geq 2|\boldsymbol{\beta}_0^T\boldsymbol{u}| \right\} \, ,\]  
and hence $\mathit{ENT}_i$ is semi-algebraic.

\paragraph{Case II: $S_i$ does not contain a positive real eigenvalue.}
It follows from~\cite[Lemma 4]{Bra06} that if $S_i$ does not contain a
positive real eigenvalue then for $\boldsymbol{u}\in V_i$ the sequence
$\boldsymbol{b}^TA^n\boldsymbol{u}$ is either identically zero for $n\geq
d$ or is infinitely often strictly positive and infinitely often
strictly negative.  Thus in this case
$\mathit{ENT}_i=\mathit{ZERO}_i$.  But we have already shown that
$\mathit{ZERO}_i$ is semi-algebraic.

\paragraph{Case III: all complex eigenvalues in $S_i$ are simple.}
Suppose that all complex eigenvalues in $S_i$ are simple.  If $S_i$
contains no positive real eigenvalue then Case II applies.  Thus we
may assume that $S_i$ comprises a positive real eigenvalue $\rho$ of
index $t$ and simple complex eigenvalues
$\lambda_1,\overline{\lambda_1},\ldots,\lambda_s,\overline{\lambda_s}$.
Given $\boldsymbol{u}\in V_i$ we can write
\begin{eqnarray}
\boldsymbol{b}^TA^n\boldsymbol{u} &=& \boldsymbol{b}^TP^{-1}J^nP\boldsymbol{u}
\notag\\
&=& \left[\sum_{j=0}^{t-1} \boldsymbol{\alpha}_j n^j \rho^n + \sum_{j=1}^s 
(\boldsymbol{\beta}_j \lambda_j^n + \overline{\boldsymbol{\beta}_j \lambda^n_j}) 
\right]^T\boldsymbol{u} \, ,
\label{eq:exp-poly}
\end{eqnarray}
where the $\boldsymbol{\alpha}_j$ and $\boldsymbol{\beta}_j$ are
$d$-dimensional vectors of algebraic numbers, with all coefficients of
each $\boldsymbol{\alpha}_j$ being real.

Since $\rho=|\lambda_1|=\ldots=|\lambda_s|$, if
$\boldsymbol{\alpha}_j^T\boldsymbol{u}\neq 0$ for some strictly
positive index $j$, then, for the largest such index $j$, the term
$n^{j}\rho^n\boldsymbol{\alpha}_{j}^T\boldsymbol{u}$ is 
dominating on the right-hand side of (\ref{eq:exp-poly}).  In
particular, if $\boldsymbol{\alpha}_{j}^T\boldsymbol{u}>0$ then the
sequence $\boldsymbol{b}^TA^n\boldsymbol{u}$ is ultimately positive
(hence $\boldsymbol{u}\in \mathit{ENT}_i$), and if
$\boldsymbol{\alpha}_{j}^T\boldsymbol{u}<0$ then
$\boldsymbol{b}^TA^n\boldsymbol{u}$ is not ultimately positive (hence
$\boldsymbol{u}\not\in\mathit{ENT}_i$).  It follows that
\begin{gather}
\left\{ \boldsymbol{u} \in V_i : \bigvee_{j=1}^{t-1}\bigwedge_{k=j+1}^{t-1}
(\boldsymbol{\alpha}_j^T \boldsymbol{u}>0 \wedge
\boldsymbol{\alpha}_k^T \boldsymbol{u}=0) \right\} 
\label{eq:ENT1}
\end{gather}
is a subset of $\mathit{ENT}_i$.

The case that $\boldsymbol{\alpha}_j^T\boldsymbol{u}=0$ for all
$j=1,\ldots,t-1$ is more subtle since there is no single dominant term
in (\ref{eq:exp-poly}); this is where we employ the results of
Section~\ref{sec:mult} on multiplicative relations. In this case we
rewrite (\ref{eq:exp-poly}) as
\begin{gather}
\boldsymbol{b}^T A^n \boldsymbol{u} =
 \rho^n f\left(\frac{\lambda^n_1}{\rho^n},\ldots,
               \frac{\lambda^n_s}{\rho^n}\right)^T \boldsymbol{u} \, ,
\label{eq:f-exp}
\end{gather}
where $f:\mathbb{T}^s\rightarrow\mathbb{R}^d$ is defined by 
\[f(z_1,\ldots,z_s)=\boldsymbol{\alpha}_0 + \sum_{j=1}^s \boldsymbol{\beta}_jz_j + \overline{\boldsymbol{\beta}_jz_j} \, .\]
Defining $\boldsymbol\mu=(\lambda_1/\rho,\ldots,\lambda_s/\rho)$, we furthermore rewrite (\ref{eq:f-exp}) as
\begin{gather}  
\boldsymbol{b}^T A^n \boldsymbol{u} =
    \rho^n f(\boldsymbol\mu^n)^T \boldsymbol{u} \, . 
\label{eq:f2-exp}
\end{gather}

By Theorem~\ref{dense}, $\{ \boldsymbol\mu^n : n \in
\mathbb{N} \}$ is a dense subset of the torus $T(\boldsymbol{\mu})$.
Thus the right-hand side of (\ref{eq:f2-exp}) is non-negative for
every $n$ if and only if $f(\boldsymbol{z})^T\boldsymbol{u}\geq 0$ for
all $\boldsymbol{z}\in T(\boldsymbol{\mu})$.  It follows that
\begin{gather} \big \{ \boldsymbol{u}\in V_i : 
\forall \boldsymbol{z}\in T(\boldsymbol{\mu}),\,
                              f(\boldsymbol{z})^T\boldsymbol{u}\geq 0 
\big\}\, . 
\label{eq:ENT2}
\end{gather}
is a subset of $\mathit{ENT}_i$.

In Section~\ref{sec:mult} we observed that the set
$T(\boldsymbol{\mu})$ was (effectively) semi-algebraic.  It follows
that we can express the condition $\forall \boldsymbol{z}\in
T(\boldsymbol{\mu}),\, f(\boldsymbol{z})^T\boldsymbol{u}\geq 0$ in the
first-order theory of the reals.  By the Tarski-Seidenberg theorem
\cite{Tar51} on quantifier elimination, the set of $\boldsymbol{u}\in
\mathbb{R}^d$ satisfying this condition is semi-algebraic.  But now
$\mathit{ENT}_i$ is the union of the two semi-algebraic sets
(\ref{eq:ENT1}) and (\ref{eq:ENT2}), and therefore $\mathit{ENT}_i$ is
itself semi-algebraic.
\end{proof}
\subsection{Definition of a Witness Set}
Having shown that $\mathit{ZERO}_i$ and $\mathit{ENT}_i$ are
semi-algebraic sets for $i=1,\ldots,m$, we now define a witness set
$W$ for the loop \textsf{P4}.  

Given $\boldsymbol{u} \in \mathbb{R}^d$, write $\boldsymbol{u}=
\boldsymbol{u}_1+\ldots+\boldsymbol{u}_m$, with $\boldsymbol{u}_1 \in
V_1,\ldots,\boldsymbol{u}_m\in V_m$.  Say that $\boldsymbol{u}_i$ is
the \emph{dominant component} of $\boldsymbol{u}$ if
$\boldsymbol{u}_i\not\in \mathit{ZERO}_i$ and
$\boldsymbol{u}_j\in\mathit{ZERO}_j$ for all $j<i$.  The intuition is
that if $\boldsymbol{u}_i$ is dominant then the eventual
non-termination of \textsf{P4} on $\boldsymbol{u}$ is determined by
its eventual non-termination on $\boldsymbol{u}_i$.  However, to prove
this we need to assume $\boldsymbol{u}\in (\mathbb{A}\cap\mathbb{R})^d$.  Formally we
have:

\begin{proposition}
  If $\boldsymbol{u}_i$ is the dominant component of $\boldsymbol{u}
  \in (\mathbb{A}\cap\mathbb{R})^d$ then $\boldsymbol{u}\in\mathit{ENT}$ if and only
if $\boldsymbol{u}_i \in \mathit{ENT}$.
\label{prop:dominant}
\end{proposition}
\begin{proof}
From the fact that $\boldsymbol{u}_i$ is dominant we have:
\begin{eqnarray}
 \boldsymbol{b}^TA^n\boldsymbol{u}&=&
   \boldsymbol{b}^TA^n(\boldsymbol{u}_1+\cdots+\boldsymbol{u}_m) \notag\\
&=&\boldsymbol{b}^TA^n(\boldsymbol{u}_i+\cdots+\boldsymbol{u}_m) 
\label{eq:part}
\end{eqnarray}
for all $n\geq d$.  Moreover, for each $j>i$ it is clear that
  $|\boldsymbol{b}^TA^n\boldsymbol{u}_j| = O(n^d\rho_j^n)$, where
  $\rho_j\geq 0$ is the modulus of the eigenvalues in $S_j$.

  We now consider three cases, mirroring the proof of
  Proposition~\ref{prop:semi-alg}.

  The first case is that $A$ has dimension at most $5$.  As observed
  in the proof of Proposition~\ref{prop:semi-alg}, all instances of
  this case that are not already covered by the second and third cases
  are such that $S_i$ contains all the eigenvalues of $A$, and hence
  $\boldsymbol{u}_i=\boldsymbol{u}$.  In this situation the proposition
  holds trivially.

  The second case is that $S_i$ does not contain a positive real
  eigenvalue.  Then it follows from~\cite[Lemma 4]{Bra06} that there
  is a constant $c<0$ such that
  $\boldsymbol{b}^TA^n\boldsymbol{u}_i<c\rho_i^n$ for infinitely many
  $n$.  In this case neither $\boldsymbol{u}_i$ nor $\boldsymbol{u}$
  are elements of $\mathit{ENT}$.

  It remains to consider the case that all complex eigenvalues in
  $S_i$ are simple.  Suppose that the dominant term in the expression
  for $\boldsymbol{b}^TA^n\boldsymbol{u}_i$ has the form $\alpha
  n^k\rho_i^n$ for some real constant $\alpha\neq 0$ and $k>0$.  If
  $\alpha>0$ then both $\boldsymbol{u}$ and $\boldsymbol{u}_i$ are in
  $\mathit{ENT}$ and if $\alpha<0$ then neither $\boldsymbol{u}$ or
  $\boldsymbol{u}_i$ are in $\mathit{ENT}$.  

  Otherwise, specialising the expression (\ref{eq:master}) to the case
  at hand, we have that
\begin{gather}
\label{eq:diag}
 \boldsymbol{b}^TA^n\boldsymbol{u}_i = \alpha_0 \rho_i^n+\sum\limits_{j=1}^s\beta_j\lambda_j^n+\overline{\beta_j\lambda_j^n}
\end{gather}
where $\alpha_0$ and the $\beta_j$ are algebraic-integer constants and
$\rho_i,\lambda_1,\overline{\lambda_1},\ldots,\lambda_s,\overline{\lambda_s} \in S_i$.  In this case one can use the
$S$-units theorem of Evertse, van der Poorten, and
Schlickewei~\cite{Evertse84,PS82} to show that for all $\varepsilon>0$ it
is the case that $\boldsymbol{b}^TA^n\boldsymbol{u}_i =\Omega\left(
  \rho_i^n \Lambda^{-n\varepsilon} \right)$,
where $\Lambda$ is an upper bound on the absolute value of eigenvalues of $A$ (see the Appendix for details).

From this lower bound, taking $\varepsilon$ suitably small, it follows
that $|\boldsymbol{b}^TA^n\boldsymbol{u}_j|=o(|\boldsymbol{b}^TA^n
\boldsymbol{u}_i|)$ for all $j>i$ and hence that
$\boldsymbol{u}\in\mathit{ENT}$ if and only if $\boldsymbol{u}_i\in
\mathit{ENT}$.
\end{proof}

Now we define a witness set $W$ for program \textsf{P4} by
\begin{align*}
W :=   \bigcup_{i=1}^m \{\boldsymbol{u}\in \mathbb{R}^d: \,
         & \boldsymbol{u}_i \mbox{ is the dominant component of } \boldsymbol{u},\\
         & \boldsymbol{u}_i \in \mathit{ENT} \} 
\cup \mathit{ZERO} \, .
\end{align*}

From the fact that $\mathit{ZERO}_i$, $\mathit{ENT}_i$, and $V_i$ are
semi-algebraic for $i=1,\ldots,m$, it is easy to see that $W$ is
semi-algebraic.  It moreover follows from Proposition~\ref{prop:dominant}
that $W \cap \mathbb{A}^d = \mathit{ENT} \cap \mathbb{A}^d$.

To conclude the proof of Proposition~\ref{prop:main2}, it 
remains to observe that the witness set $W$, like the actual set 
$\mathit{ENT}$ of eventually non-terminating points, is convex.
\begin{proposition}
The witness set $W$ is convex.
\end{proposition}
\begin{proof}
  Suppose $\boldsymbol y,\boldsymbol z\in W$ and let $\boldsymbol
  x=\lambda\boldsymbol y+(1-\lambda)\boldsymbol z$, where $0 < \lambda
  < 1$. Moreover, write
  $\boldsymbol{x}=\boldsymbol{x}_1+\ldots+\boldsymbol{x}_m$, where
  $\boldsymbol{x}_1\in V_1,\ldots,\boldsymbol{x}_m\in V_m$, and
  likewise for $\boldsymbol{y}$ and $\boldsymbol{z}$.

If $\boldsymbol{y},\boldsymbol{z}\in \mathit{ZERO}$ then
$\boldsymbol{x}\in \mathit{ZERO}$ since the latter is a convex set.

Suppose that $\boldsymbol{y}\in \mathit{ZERO}$ and $\boldsymbol{z}_i \in \mathit{ENT}$ is dominant for $\boldsymbol{z}$ for some index $i\in\{1,\ldots,m\}$.
Then $\boldsymbol{x}_i$ is dominant for $\boldsymbol{x}$, and
$\boldsymbol{x}_i \in \mathit{ENT}$.  Thus $\boldsymbol{x}\in W$.

Otherwise, let $\boldsymbol y_i$ be dominant for $\boldsymbol y$ and
$\boldsymbol z_j$ be dominant for $\boldsymbol z$ for some
$i,j\in\{1,\ldots,m\}$.  Then $\boldsymbol x_k\in \mathit{ZERO}_k$ for
all $k<\min\lbrace i,j\rbrace$ since $\mathit{ZERO}_k$ is convex.
Moreover if $k=\min\lbrace i,j\rbrace$ then $\boldsymbol
y_k,\boldsymbol z_k \in \mathit{ENT}_k$, and hence $\boldsymbol x_k
\in \mathit{ENT}_k$ by convexity of $\mathit{ENT}_k$.  It follows that
$\boldsymbol{x}\in W$.
\end{proof}

This concludes the proof of Proposition~\ref{prop:main2}.
In the remaining part of this section we show that
$\overline{\mathit{ENT}}=\overline{W}$.  

The inclusion $\overline{W}\subseteq \overline{ENT}$ can be shown
using the fact that the set algebraic points in any semi-algebraic set
is dense in that set. (See the Appendix for details).  From
this we have: 
\begin{align*}
\overline{W}=\overline{W\cap\mathbb{A}^d}=\overline{\mathit{ENT}\cap\mathbb{A}^d} \subseteq \overline{\mathit{ENT}}\cap\overline{\mathbb{A}^d}=\overline{\mathit{ENT}}
\end{align*}

The reverse inclusion, $\overline{ENT}\subseteq\overline{W}$, can be
shown in similar fashion but this time using the fact that
$\mathit{ENT}\cap\mathbb{A}^d$ is dense in $\mathit{ENT}$.  Our
remaining goal is this last fact, which is established in
Corollary~\ref{corl:dense} below.

We have previously shown that a vector of algebraic numbers
$\boldsymbol{u}\in (\mathbb{A}\cap\mathbb{R})^d$ is eventually
non-terminating if and only if its dominant component
$\boldsymbol{u}_i$ is eventually non-terminating.  We now prove a
partial result of this nature for general vectors $\boldsymbol{u}\in
\mathbb{R}^d$.

\begin{proposition}
  Suppose that $\boldsymbol u=\boldsymbol u_1+\cdots+\boldsymbol
  u_m\in \mathbb{R}^d$, where $\boldsymbol u_1\in
  V_1,\ldots,\boldsymbol u_m\in V_m$. Then $\boldsymbol
  u\in\mathit{ENT}$ implies that its dominant component $\boldsymbol
  u_i$ is also in \textit{ENT}.
\label{prop:dom2}
\end{proposition}
\begin{proof}
  The only non-trivial case corresponds to the situation in which
  $\boldsymbol b^T A^n \boldsymbol u_i$ is of the form (\ref{eq:diag}).
  Let $f$ and $\boldsymbol \mu$ be as in (\ref{eq:f-exp}), that is, so
  that $\boldsymbol b^T A^n \boldsymbol u_i=\rho_i^n f (\boldsymbol
  \mu^n)^T \boldsymbol u_i$. If $\boldsymbol u_i\not\in\mathit{ENT}$,
  then there exists some constant $c<0$ and some $\boldsymbol z\in T(\boldsymbol \mu)$ such that $f(\boldsymbol z)^T \boldsymbol u_i=c$. Therefore, for any $\varepsilon>0$, $\boldsymbol b^T A^n \boldsymbol u_i<(c+\varepsilon)\rho_i^n$ holds for infinitely many $n$, due to Proposition~\ref{dense} and to continuity of $f$, and so $\boldsymbol u\not\in\mathit{ENT}$.
\end{proof}

\begin{corollary}
$\mathit{ENT}\cap\mathbb{A}^d$ is dense in \textit{ENT}.
\label{corl:dense}
\end{corollary}

\begin{proof}
  At several points we will rely on the fact that if
  $X\subseteq\mathbb{R}^d$ is semi-algebraic, then the algebraic
  points in $X$ are dense in $X$.  (See Appendix for details.)

  Fix $\boldsymbol u\in\mathit{ENT}$ and let $\varepsilon>0$ be
  given. We will find $\boldsymbol v\in \mathit{ENT}\cap \mathbb{A}^d$
  such that $||\boldsymbol u-\boldsymbol v||<\varepsilon$.
 
  The case in which $\boldsymbol u \in \mathit{ZERO}$ is easy since
  $\mathit{ZERO}$ is semi-algebraic and so we can take $\boldsymbol v$
  to be an algebraic point in $\mathit{ZERO}$ that is suitably close
  to $\boldsymbol u$.

  Suppose now that $\boldsymbol u=\boldsymbol u_1+\cdots+\boldsymbol
  u_m$, where $\boldsymbol u_1\in V_1,\ldots,\boldsymbol u_m\in V_m$,
  with $\boldsymbol u_i$ the dominant component of $\boldsymbol u$. By
  Proposition~\ref{prop:dom2}, $\boldsymbol u\in\mathit{ENT}$ implies
  that $\boldsymbol u_i\in \mathit{ENT}$. Since $\mathit{ENT}\cap V_i$
  is semi-algebraic, we can pick $\boldsymbol v_i\in \mathit{ENT}\cap
  V_i\cap\mathbb{A}^d$ such that $\| \boldsymbol v_i -\boldsymbol u_i
  \|<\varepsilon/n$. For each $j>i$, we pick some $\boldsymbol v_j\in
  V_j\cap \mathbb{A}^d$ for which $\| \boldsymbol v_j - \boldsymbol
  u_j\|<\varepsilon/n$.  For each $j<i$ we pick some $v_j \in
  \mathit{ZERO}\cap V_j\cap \mathbb{A}^d$ for which $\| \boldsymbol
  v_j - \boldsymbol u_j\|<\varepsilon/n$.

  Then, letting $\boldsymbol v=\boldsymbol v_1+\cdots+\boldsymbol v_m
  \in \mathbb{A}^d$, it follows that $\|\boldsymbol u-\boldsymbol
  v\|<\varepsilon$.  Finally, by Proposition~\ref{prop:dominant} we
  have $\boldsymbol v\in\mathit{ENT}$ since $\boldsymbol v_i$ is the
  dominant component of $\boldsymbol v$ and $\boldsymbol v_i \in
  \mathit{ENT}$ by construction.
\end{proof}
\section{Complexity Analysis}
\label{sec:complexity}

The purpose of this section is to justify our previous claims about
the complexity of the algorithm presented in this paper. We do this by
proving the following result.

\begin{proposition}
Our procedure requires space $\mathit{poly}(\log\max_{i,j}\lvert A_{ij}\rvert,d)^{\mathit{poly}(d)}$.
\end{proposition}

\begin{proof}
There are three critical steps in our procedure for which a
super-polynomial amount of space is required: when reducing to the
case in which $A$ is non-degenerate, when performing quantifier
elimination, and when testing whether the witness set $W$ intersects
the integer lattice.

The last of these steps runs in space $SD^{O(d^4)}$, where $S$ denotes
the size of the representation of the quantifier-free formula defining
the witness set $W$, $D$ denotes the maximum degree of the polynomials
occurring in that formula, and $d$ denotes the dimension of the
ambient space. Since $d$ remains fixed throughout the procedure (apart
from an increase by $1$ in the reduction to the homogeneous case), it
remains to show that $S$ and $D$ are bounded by an expression of the
form $\mathit{poly}(\log\max_{i,j}\lvert
A_{ij}\rvert,d)^{\mathit{poly}(d)}$.

The reduction to the case in which $A$ is non-degenerate entails an
increase by a factor of $\mathit{poly}(\log\max_{i,j}\lvert
A_{ij}\rvert,d)^{\mathit{poly}(d)}$ in the size of the formula
defining the witness set $W$, as the least common multiple of the
orders of all ratios of eigenvalues of $A$ that are roots of unity is
$L=2^{O(d\sqrt{\log d})}$ and $\log\max_{i,j}\lvert A^L_{ij}\rvert\leq
\log (d^L\max_{i,j}\lvert A_{ij}\rvert^L)=L\log (d\max_{i,j}\lvert
A_{ij}\rvert)$.

It remains to show that the quantifier-free formula defining the witness set $W$ in the case where $A$ is non-degenerate takes space $\mathit{poly}(\log\max_{i,j}\lvert A_{ij}\rvert,d)^{\mathit{poly}(d)}$ and involves exclusively polynomials of degree $\mathit{poly}(\log\max_{i,j}\lvert A_{ij}\rvert,d)^{\mathit{poly}(d)}$.

Let $D_0,H_0$ denote the maximum degree and height across all the eigenvalues of $A$, respectively. Then $D_0\leq d$ and $\log H_0\leq\log (d!\max_{i,j}\lvert A_{ij}\rvert^d)\leq d\log (d\max_{i,j}\lvert A_{ij}\rvert)$. Before performing quantifier elimination, the degree of any polynomial in the defining formula of the witness set $W$ is bounded by $(D_0\log H_0)^{O(d^2)}$, and the number of such polynomials is bounded by $O(d)$, by Masser's theorem. Finally, after applying quantifier elimination, we know that $D\leq(D_0\log H_0)^{O(d^3)}$ and that $S\leq d^{O(d^2)}(D_0\log H_0)^{O(d^4)}$, thanks to Theorem~\ref{thm:quant-elim}.
\end{proof}

\section{Conclusion}

We have shown decidability of termination of simple linear loops over
the integers under the assumption that the update matrix is
diagonalisable, partially answering an open problem
of~\cite{Tiw04,Bra06}.  As we have explained before, the termination
problem on the same class of linear loops, but for fixed initial
values, seems to have a different character and to be more difficult.
In this respect it is interesting to note that there are other
settings in which universal termination is an easier problem than
pointwise termination. For example, universal termination of Petri
nets (also known as \textit{structural boundedness}) is
\textbf{PTIME}-decidable, but the pointwise termination problem is
\textbf{EXPSPACE}-hard.

A natural subject for further work is whether our techniques can be
extended to non-diagonalisable matrices, or whether, as is the case
for pointwise termination~\cite{OW14:SODA}, there are unavoidable
number-theoretic obstacles to proving decidability.  We would also like
to further study the computational complexity of the termination
problem.  While there is a large gap between the \textbf{coNP} lower
complexity bound mentioned in the Introduction and the exponential
space upper bound of our procedure, this may be connected with the
fact that our procedure computes a representation of the set of all
integer eventually non-terminating points. Finally we would like to
examine more carefully the question of whether the respective sets of
terminating and non-terminating points are semi-algebraic.  Note that
an \emph{effective} semi-algebraic characterisation of the set of
terminating points would allow us to solve the termination problem
over fixed initial values.



\section*{Acknowledgements}

The authors would like to thank Elias Koutsoupias and Ventsislav
Chonev for their advice and feedback.

\newpage
\section{Appendix}

\subsection{Algebraic Numbers}

The purpose of this section is threefold: to introduce the main concepts in Algebraic Number Theory, necessary to understanding the hypothesis for the $S$-units theorem, stated below; to justify the application of the aforementioned result in lower-bounding the dominant terms of linear recurrence sequences; to explain how one can effectively manipulate algebraic numbers.

\subsection{Preliminaries}

A complex number $\alpha$ is said to be \textbf{algebraic} if it is the root of some polynomial with integer coefficients. Among those polynomials, there exists a unique one of minimal degree whose coefficients have no common factor, and it is said to be the \textbf{defining polynomial} of $\alpha$, denoted by $p_\alpha$, and it is always an irreducible polynomial. Moreover, if $p_\alpha$ is monic, $\alpha$ is said to be an \textbf{algebraic integer}. The degree of an algebraic number is defined as the degree of $p_\alpha$, and its height as the maximum absolute value of the coefficients of $p_\alpha$ (also said to be the height of that polynomial). The roots of $p_\alpha$ are said to be the \textbf{Galois conjugates} of $\alpha$. We denote the set of algebraic numbers by $\mathbb{A}$, and the set of algebraic integers by $\mathcal{O}_\mathbb{A}$. For all $\alpha\in\mathbb{A}$, there exists some $n\in\mathbb{N}$ such that $n\alpha\in\mathcal{O}_\mathbb{A}$. It is well known that $\mathbb{A}$ is a field and that $\mathcal{O}_\mathbb{A}$ is a ring.

A \textbf{number field} of dimension $d$ is a field extension $K$ of $\mathbb{Q}$ whose degree as a vector-space over $\mathbb{Q}$ is $d$. In particular, $K\subseteq\mathbb{A}$ must hold. Recall that, in that case, there are exactly $d$ monomorphisms $\sigma_i:K\rightarrow\mathbb{C}$ whose restriction over $\mathbb{Q}$ is the identity (and therefore these must map elements of $K$ to their Galois conjugates). The \textbf{ring of integers} $\mathcal{O}$ of a number field $K$ is the set of elements of $K$ that are algebraic integers, that is, $\mathcal{O}=K\cap\mathcal{O}_\mathbb{A}$. An ideal of $\mathcal{O}$ is an additive subgroup of $\mathcal{O}$ that is closed under multiplication by any element of $\mathcal{O}$. An ideal $\mathfrak{P}$ is said to be prime if $ab\in \mathfrak{P}$ implies $a\in\mathfrak{P}$ or $b\in\mathfrak{P}$. The following theorem is central in Algebraic Number Theory \cite{SnT}:

\begin{theorem}
In any ring of integers, ideals can be uniquely factored as products of prime ideals up to permutation.
\end{theorem}

\subsection{Lower-bounding simple linear recurrence sequences}

We are interested in lower-bounding expressions of the form
\begin{equation}
\label{eq:sum}
u_n=\sum\limits_{j=1}^s\alpha_j\lambda_j^n
\end{equation}
where the $\alpha_j$ are algebraic-integer constants and $\lambda_1,\ldots,\lambda_s$ have the same absolute value $\rho$. Any such sequence must in fact be a simple linear recurrence sequence with algebraic coefficients and characteristic roots $\lambda_1,\ldots,\lambda_s$, as explained in Section 1.1.6 of \cite{BOOK}.

The next theorem, by Evertse, van der Poorten, and Schlickewei, was established in \cite{Evertse84,PS82} to
analyse the growth of linear recurrence sequences. It gives us a very strong lower bound on the magnitude of sums of $S$-units, as defined below. Its key ingredient is Schlickewei's $p$-adic generalisation \cite{Sch77} of Schmidt's
Subspace theorem.

Let $S$ be a finite set of prime ideals of the ring of integers
$\mathcal{O}$ of a number field $K$. We say that
$\alpha\in\mathcal{O}$ is an \textbf{$S$-unit} if all the ideals
appearing in the prime factorisation of $(\alpha)$, that is, the ideal
generated by $\alpha$, are in $S$. 

\begin{theorem}[$S$-units]
\label{thm:s-units}
Let $K$ be a number field, $s$ be a positive integer, and $S$ be a
finite set of prime ideals of $\mathcal{O}$. Then for every
$\varepsilon>0$ there exists a constant $C$, depending only on $s$,
$K$, $S$, and $\varepsilon$, with the following property. For every
set of $S$-units $x_1,\ldots,x_s\in\mathcal{O}$ such that
$\sum\limits_{i\in I} x_i\neq 0$ for all non-empty $I\subseteq\lbrace
1,\ldots,s\rbrace$, it holds that
\[ \lvert x_1+\cdots+x_s \rvert\geq CYZ^{-\varepsilon} \]
where $Y=\max\lbrace \lvert x_j\rvert : 1\leq j\leq s \rbrace$ and $Z=\max\lbrace \sigma_i(x_j): 1\leq j\leq s,1\leq i\leq d \rbrace$ and $\sigma_i$ represent the different monomorphisms from $K$ to $\mathbb{C}$.
\end{theorem}

In order to make use of this result, it is important to understand the set
\begin{equation}
\lbrace n\in\mathbb{N}: \exists I\subseteq \lbrace 1,\ldots,s\rbrace, \sum\limits_{j\in I}\alpha_j\lambda_j^n=0\rbrace
\label{eq:genzeros}
\end{equation}

The following well-known theorem characterises the set of zeros of
linear recurrence sequences. In particular, it gives us a sufficient
condition for guaranteeing that the set of zeros of a non-identically
zero linear recurrence sequence is finite. Namely, it suffices that the sequence is non-degenerate, that is, that no ratio of two of its characteristic roots is a root of unit.

\begin{theorem}[Skolem-Mahler-Lech]
Let $u_n=\sum\limits_{j=1}^l \alpha_j\lambda_j^n$ be a linear recurrence sequence. The set $\lbrace n\in\mathbb{N}: u_n=0\rbrace$ is always a union of a finite set and finitely many arithmetic progressions. Moreover, if $u_n$ is non-degenerate, this set is actually finite.
\end{theorem}

Therefore, it follows from the Skolem-Mahler-Lech theorem that if $u_n$ is non-degenerate  then (\ref{eq:genzeros}) must be finite, assuming without loss of generality that $\sum\limits_{j\in I}\alpha_j\lambda_j^n$ is never eventually zero.

We can now apply the $S$-units theorem in order to get a lower bound on (\ref{eq:sum}) that holds for all but finitely many $n$, by letting $K$ be the splitting field of the characteristic polynomial of $u_n$, $S$ be the set of prime ideals of the ring of integers of $K$ that appear in the factorisation of each of the algebraic integers $\alpha_j$ and $\lambda_j$, and $x_j=\alpha_j\lambda_j^n$ for each $j$, making (\ref{eq:sum}) a sum of $S$-units.

In the notation of the theorem, we have $Y=\Omega(\rho^n)$. If $\Lambda$ is an upper bound on the absolute value of the Galois conjugates of each $\lambda_j$ (that is, each $\sigma_i(\lambda_j)$), then $Z=O(\Lambda^n)$. Thus, for any $\varepsilon>0$, we know that
\begin{align*}
\sum\limits_{j=1}^s\alpha_j\lambda_j^n &=\Omega(YZ^{-\varepsilon})=
\Omega\left(\rho^n\Lambda^{-n\varepsilon}\right)
\end{align*}
Finally, we note that by picking $\varepsilon$ to be sufficiently small we can get $\rho\Lambda^{-\varepsilon}$ arbitrarily close to $\rho$.

\subsection{Manipulating algebraic numbers}

The following separation bound allows us to effectively represent an arbitrary algebraic number by keeping its defining polynomial, a sufficiently accurate estimate for the root we want to store, and an upper bound on the error. We call this its \textbf{standard/canonical representation}.

\begin{lemma}[Mignotte]
Let $f\in\mathbb{Z}[x]$. Then
\begin{equation}
f(\alpha_1)=0=f(\alpha_2)\Rightarrow \lvert \alpha_1-\alpha_2\rvert>\frac{\sqrt{6}}{d^{(d+1)/2}H^{d-1}}
\end{equation}
where $d$ and $H$ are respectively the degree and height of $f$.
\end{lemma}

It is well known that arithmetic operations and equality testing on these numbers can be done in polynomial time on the size of the canonical representations of the relevant numbers, since one can:
\begin{itemize}
\item compute polynomially many bits of the roots of any polynomial $p\in\mathbb{Q}[x]$ in polynomial time, due to the work of Pan in \cite{Pan97}
\item find the minimal polynomial of an algebraic number by factoring the polynomial in its description in polynomial time using the LLL algorithm \cite{LenstraLenstraLovasz1982}
\item use the sub-resultant algorithm (see Algorithm 3.3.7 in \cite{Coh93}) and the two aforementioned procedures to compute canonical representations of sums, differences, multiplications, and divisions of canonically represented algebraic numbers
\end{itemize}

Moreover, we need to know how to decide whether a given canonically represented algebraic number $\alpha$ is a root of unity, that is, whether $\alpha^r=1$ for some $r$. If that is the case, then its defining polynomial will be the $r$-th cyclotomic polynomial, which has degree $\phi(r)$, if $r$ is taken to be minimal, that is, if $\alpha$ is a primitive $r$-th root of unity. The following (crude) lower bound on $\phi(r)$ allows us to decide this in polynomial time, assuming that the degree of $\alpha$ is given in unary.

\begin{lemma}
Let $\phi$ be Euler's totient function. Then $\phi(r)\geq\sqrt(r/2)$. Therefore, if $\alpha$ has degree $n$ and is a $r$-th root of unity, then $r\leq 2n^2$.
\end{lemma}

Therefore, in order to decide whether an algebraic number $\alpha$ of degree $n$ is a root of unity, we check whether it is a $r$-th root of unity, for each $r\leq 2n^2$. In order to test whether $\alpha$ is a $r$-th root of unity, it suffices to see whether $\mathit{gcd}(p_\alpha,x^r-1)=p_\alpha$, since we know that $x^r-1$ is the product of each $d$-th cyclotomic polynomial, with $d$ ranging over the divisors of $n$.
\subsection{First-Order Theory of Reals}

Let $\boldsymbol{x}=(x_1,\ldots,x_m)$ be a list of $m$ real-valued
variables, and let $\sigma(\boldsymbol{x})$ be a Boolean combination
of atomic predicates of the form $g(\boldsymbol{x})\sim 0$, where each
$g(\boldsymbol{x})$ is a polynomial with integer coefficients in the
variables $\boldsymbol{x}$, and $\sim$ is either $>$ or $=$. Tarski
has famously shown that we can decide the truth over the field
$\mathbb{R}$ of sentences of the form $\phi=Q_1 x_1 \cdots Q_m x_n
\sigma(\boldsymbol{x})$, where $Q_i$ is either $\exists$ or
$\forall$. He did so by showing that this theory admits quantifier
elimination (Tarski-Seidenberg theorem \cite{Tar51}).


All sets that are definable in the first-order theory of reals without
quantification are by definition semi-algebraic, and it follows from
Tarski's theorem that this is still the case if we allow
quantification. We also remark that our standard representation of
algebraic numbers allows us to write them explicitly in the
first-order theory of reals, that is, given $\alpha\in\mathbb{A}$,
there exists a sentence $\sigma(x)$ such that $\sigma(x)$ is true if
and only if $x=\alpha$. Thus, we allow their use when defining
semi-algebraic sets, for simplicity.

It follows from the undecidability of Hilbert's Tenth Problem that, in
general, we cannot decide whether a given semi-algebraic set has an
integer point.

We shall make use of the following result by Basu, Pollack, and Roy
\cite{BasuPR96}, which tells us how expensive it is, in terms of space
usage, to perform quantifier elimination on a formula in the
first-order theory of reals:

\begin{theorem}
  \label{thm:quant-elim}
  Given a set $\mathcal{Q}=\lbrace q_1,\ldots,q_s\rbrace$ of $s$
  polynomials each of degree at most $D$, in $h+d$ variables, and a
  first-order formula $\Phi(\boldsymbol x)=Q y_1 \ldots Q y_h
  F(q_1(\boldsymbol x,\boldsymbol y),\ldots,q_s(\boldsymbol
  x,\boldsymbol y))$, where $Q\in\lbrace \exists,\forall\rbrace$, $F$
  is a quantifier-free Boolean combination with atomic elements of the
  form $q_i(\boldsymbol x,\boldsymbol y)\sim 0$, then there exists a
  quantifier-free formula $\Psi(\boldsymbol
  x)=\bigwedge_{i=1}^J\bigvee_{j=1}^{J_i}q_{ij}(\boldsymbol x)\sim 0$,
  where $I\leq (sD)^{O(hd)}$, $J\leq (sD)^{O(d)}$, the degrees of the
  polynomials $q_{ij}$ are bounded by $D^d$, and the bit-sizes of the
  heights of the polynomials in the quantifier-free formula are only
  polynomially larger than those of $q_1,\ldots,q_s$.
\end{theorem}

We also make use of the following lemmas:
\begin{lemma}
  If $X\subseteq\mathbb{R}^d$ is semi-algebraic and non-empty,
  $X\cap\mathbb{A}^d\neq\emptyset$.
\end{lemma}

\begin{proof}
  We prove this result by strong induction on $d$. Since $X$ is
  semi-algebraic, there exists a quantifier-free sentence in the
  first-order theory of reals $\sigma$ such that $X=\lbrace
  x\in\mathbb{R}^d\mid \sigma(x)\rbrace$.

  Suppose that $d>1$. Letting $X_1=\lbrace x_d\in\mathbb{R}\mid
  \exists
  x_1,\ldots,x_{d-1}\in\mathbb{R}^{d-1},\sigma(x_1,\ldots,x_d)\rbrace$
  and since $X_1\neq\emptyset$ is semi-algebraic, by the induction
  hypothesis, there must be $x_d^*\in\mathbb{A}\cap X_1$. Moreover, we
  can define $X_2=\lbrace (x_2,\ldots,x_d)\in\mathbb{R}^{d-1}\mid
  \sigma(x_1^*,x_2,\ldots,x_n)\rbrace$, which is non-empty and
  semi-algebraic, and again by induction hypothesis there exists some
  $(x_2^*,\ldots,x_d^*)\in\mathbb{A}^{d-1}\cap X_2$.

  It remains to prove this statement for $d=1$. When $d=1$, $X$ must
  be a finite union of intervals and points, since semi-algebraic sets
  form an o-minimal structure on $\mathbb{R}$ \cite{Tar51}. Clearly
  $\mathbb{A}$ is dense in any interval, and each of these isolated
  points $x$ corresponds to some constraint $g(x)=0$, which implies
  that $x$ must be algebraic, since $g$ has integer coefficients.
\end{proof}

\begin{lemma}
If $X\subseteq\mathbb{R}^d$ is semi-algebraic, then $X\cap\mathbb{A}^d$ is dense in $X$.
\end{lemma}

\begin{proof}
  Pick $x\in X$ and $\varepsilon>0$ arbitrarily. Let
  $y\in\mathbb{Q}^d$ be such that $\| x-y \|<\varepsilon/2$. Since
  $B(y,\varepsilon/2)$ is semi-algebraic, so must be $X\cap
  B(y,\varepsilon/2)$, and so this set must contain an algebraic
  point, since it is nonempty ($x$ is in it), and that point must
  therefore be at distance at most $\varepsilon$ of $x$, by the
  triangular inequality. By letting $\varepsilon\rightarrow 0$, we get
  a sequence of algebraic points which converges to $x$.
\end{proof}

\begin{lemma}
If $X\subseteq\mathbb{R}^d$ is semi-algebraic, so is $\overline{X}$.
\end{lemma}

\begin{proof}
  Let $\sigma$ be a sentence in the first-order theory of reals such
  that $X=\lbrace x\in\mathbb{R}^d\mid \sigma(x)\rbrace$. Whence
\begin{equation*}
  \overline{X}=\lbrace x\in\mathbb{R}^d\mid 
\forall \varepsilon>0,\exists y\in\mathbb{R}^d,\sigma(y)\wedge y\in B(x,\varepsilon) \rbrace .
\end{equation*}
\end{proof}

\newpage
\bibliographystyle{plain}
\bibliography{refs}

\end{document}